\documentclass{elsarticle}

\RequirePackage[OT1]{fontenc}
\RequirePackage{amsthm,amsmath}
\RequirePackage[numbers]{natbib}
\RequirePackage[colorlinks,citecolor=blue,urlcolor=blue]{hyperref}
\RequirePackage{mathrsfs} 
\RequirePackage{mathptmx}
\usepackage{float}
\RequirePackage{latexsym}

\RequirePackage{amsmath, amssymb,amsfonts}

\RequirePackage{natbib}

\RequirePackage{color}

\RequirePackage{url}
\usepackage{apptools}

\RequirePackage{algorithm}
\RequirePackage{algpseudocode}
\algnewcommand\algorithmicinput{\textbf{Input:}}
\algnewcommand\INPUT{\item[\algorithmicinput]}

\theoremstyle{plain}
\newtheorem{thm}{Theorem}

\newtheorem{lem}{Lemma}
\newtheorem{rem}{Remark}

\usepackage{lineno,hyperref}

\journal{Journal of \LaTeX\ Templates}

\begin{document}

\begin{frontmatter}

\title{Black-box sampling for weakly smooth Langevin Monte Carlo using $p$-generalized Gaussian smoothing}

\author{Anh D Doan, Xin Dang}
\author{Dao Nguyen\fnref{myfootnote}}
\address{Departments of Mathematics, University of Mississippi, Oxford, MS 38655, USA}

\ead{dxnguyen@olemiss.edu}

\begin{abstract}
Discretization of continuous-time diffusion processes is a widely recognized method for sampling. However, the canonical Euler-Maruyama discretization of the Langevin diffusion process, also named as Langevin Monte Carlo (LMC), studied mostly in the context of smooth (gradient-Lipschitz) and strongly log-concave densities, a significant constraint for its deployment in many sciences, including computational statistics and statistical learning. In this paper, we establish several theoretical contributions to the literature on such sampling methods. Particularly, we generalize the Gaussian smoothing, approximate the gradient using p-generalized Gaussian smoothing and take advantage of it in the context of black-box sampling. We first present a non-strongly concave and weakly smooth black-box LMC algorithm, ideal for practical applicability of sampling challenges in a general setting. 
\end{abstract}


\end{frontmatter}

\section{Introduction}
Over the last decade, Bayesian inference has become one of the most prevalent
inferring instruments for a variety of disciplines, including the computational
statistics and statistical learning \cite{robert2013monte}.
In general, Bayesian inference seeks to generate samples of the posterior distribution of the form:
\[
\pi(\mathrm{x})=\mathrm{e}^{-U(x)}/\int_{\mathbb{R}^{d}}\mathrm{e}^{-U(y)}\mathrm{d}y,
\]
where the function $U(\mathrm{x})$, also known as the potential function, is convex. The most conventional approaches, random walks Metropolis Hasting \cite{robert2013monte}, often struggle to select a proper proposal distribution for sampling. As a result, it has been proposed to consider continuous dynamics which inherently leave the objective distribution $\pi$ invariant. Probably one of the most well-known example of these continuous dynamic applications is the over-damped Langevin diffusion \cite{dalalyan2019user} associated with $U$, and its Euler-Maruyama discretization defines by 
\begin{equation}
\mathrm{x}_{k+1}=\mathrm{x}_{k}-\eta_{k}\nabla U(\mathrm{x}_{k})+\sqrt{2\eta}\xi_{k},\label{eq:2}
\end{equation}
where $(\eta_{k})_{k\geq1}$ is a sequence of step sizes which can be kept constant or decreasing to $0$, and $\xi_{k}\sim\mathcal{N}(0,\ I_{d\times d})$ are independent Gaussian random vectors. Nonetheless, there is a critical gap in the theory of discretization of an underlying stochastic differential equation (SDE) to the potential broad spectrum of applications in statistical inference. In particular, the application of techniques from SDEs traditionally requires $U(\mathrm{x})$ to have Lipschitz-continuous gradients. This requirement prohibits many typical utilizations \cite{durmus2018efficient}. 
\cite{chatterji2019langevin} has recently established an original approach to deal
with weakly smooth (possibly non-smooth) potential problems directly.
Their technique rests on results obtained from the optimization community, perturbing a gradient evaluating point by a Gaussian. In this paper, we show that Gaussian smoothing can be generalized to $p$-generalized Gaussian smoothing, which is novel in both Bayesian and optimization communities.
In a more general context, this novel perturbation process provides additional features, comparable to Gaussian smoothing while preserving the same result when a $p$-generalized Gaussian is a Gaussian. In particular, the bias and variance are bounded for a general $p$ while these bounds are equivalent to Gaussian smoothing for the case $p=2$. In addition, we also improve the convergence result of \cite{chatterji2019langevin} in Wasserstein distance by a simpler approach. 
The paper is organized as follows. Section 2 sets out the notation and context necessary to give our core theorems, generalizing the results obtained from \cite{chatterji2019langevin} for Gaussian smoothing. 
Section 3 broadens these outcomes of stochastic approximation of the potential (negative log-density) and composite structure of the potential, while Section 4 shows our conclusions.


\section{Weakly smooth LMC using $p$-generalized Gaussian smoothing}
The objective is to sample from a distribution $\pi\propto\exp(-U(\mathrm{x}))$,
where $\mathrm{x}\in\mathbb{R}^{d}$. We furnish the space $\mathbb{R}^{d}$
with the regular Euclidean norm $\Vert\cdot\Vert=\Vert\cdot\Vert_{2}$
and use $\left\langle\ ,\ \right\rangle $ to specify inner products.
While sampling from the exact distribution $\pi(\mathrm{x})$ is
generally computationally demanding, it is largely adequate to sample
from an approximated distribution $\tilde{\pi}(\mathrm{x})$ which is
in the vicinity of $\pi(\mathrm{x})$ by some distances. In this
paper, we use Wasserstein distance and briefly define it in Appendix A. 
Suppose the following conditions, identical to the first two assumptions
of \cite{chatterji2019langevin}:

A.1 \label{A1}
$U$ is convex and sub-differentiable. Specifically, for all
$\mathrm{x}\in\mathbb{R}^{d}$, there exists a sub-gradient of $U,\ \nabla U(\mathrm{x})\in\partial U(\mathrm{x})$,
ensuring that $\forall\mathrm{y}\in\mathbb{R}^{d}:U(\mathrm{y})\geq U(\mathrm{x})+\left\langle \nabla U(\mathrm{x}),y-x\right\rangle.$

A.2 \label{A2}
There exist $L<\infty$ and $\alpha\in[0,1]$ so that $\forall\mathrm{x},\ \mathrm{y}\in\mathbb{R}^{d}$,
we obtain 
\begin{equation}
\Vert\nabla U(\mathrm{x})-\nabla U(\mathrm{y})\Vert_{2}\leq L\Vert\mathrm{x}-\mathrm{y}\Vert_{2}^{\alpha},\label{eq:3}
\end{equation}
where $\nabla U(\mathrm{x})$ represents an arbitrary sub-gradient
of $U$ at $\mathrm{x}$.

\begin{rem}
Condition \ref{A2} is known as $(L,\alpha)$-weakly smoothness or Holder continuity of the (sub)gradients of $U$. A feature that follows straightfowardly from Eq.\eqref{eq:3} is that: 
\begin{equation}
U({\mathrm{y})\leq U(\mathrm{x})+\left\langle \nabla U(\mathrm{x}),\ \mathrm{y}-\mathrm{x}\right\rangle +\frac{L}{1+\alpha}\Vert\mathrm{y}-\mathrm{x}\Vert^{1+\alpha},\ \forall\mathrm{x},\ \mathrm{y}\in\mathrm{R^{d}}}.\label{eq:4}
\end{equation}
When $\alpha=1$, it is equivalent to the standard smoothness (Lipschitz
continuity of the gradients), whereas at the opposite extreme, when
$\alpha=0,\ U$ is (possibly) non-smooth and Lipschitz-continuous.
\end{rem}
\cite{chatterji2019langevin} perturbs the gradient evaluating point by a
Gaussian so that they can utilize weakly smooth Lipschitz potentials directly
without any additional structure. Here, we generalize their approach
to perturbation using $p$-generalize Gaussian smoothing. Particularly,
consider $\mu\geq0$, $p$-generalized Gaussian smoothing $U_{\mu}$
of $U$ is defined as: 
$$U_{\mu}(\mathrm{y}):=\mathrm{E}_{\xi}[U(\mathrm{y}+\mu\xi)]=\frac{1}{\kappa}\int_{\mathbb{R}^{d}}U(\mathrm{y}+\mu\xi)e^{-\left\Vert \xi\right\Vert _{p}^{p}/p}\mathrm{d}\xi,
$$
where $\kappa\stackrel{_{def}}{=}\int_{\mathbb{R}^{d}}e^{-\left\Vert \xi\right\Vert _{p}^{p}/p}\mathrm{d}\xi=\frac{2^{d}\Gamma^{d}(\frac{1}{p})}{p^{d-\frac{d}{p}}}$ and $\xi\sim N_{p}(0,I_{d\times d})$ (the $p$-generalized Gaussian distribution).
The rationale
for taking into account the $p$-generalized Gaussian smoothing $U_{\mu}$
rather than $U$ is that it typically benefits from superior smoothness properties. In particular, $U_{\mu}$ is smooth albeit $U$ is not. In addition, $p$-generalized Gaussian smoothing is more generalized than Gaussian smoothing in the sense that it contains all normal distribution when $p=2$ and all Laplace distribution when $p=1$. 
This family of distributions allows for tails that are either heavier or lighter than normal and in the limit as well as containing all the continuous uniform distribution. It is possible to study 
$p \in \mathbb{R}^+$ but to simplify the proof, we are only interested in $p \in [1,2]$. 
Here we examine some primary features of
$U_{\mu}$ based on adapting those results of \cite{chatterji2019langevin}. 
\begin{lem}
\label{2.1} If $U:\mathbb{R}^{d}\rightarrow\mathbb{R}$ is
a convex function which satisfies Eq.\eqref{eq:3} for some $L<\infty$
and $\alpha\in[0,1]$, then:

(i) $\forall\mathrm{x}\in\mathbb{R}^{d}$ : $|U_{\mu}({\displaystyle \mathrm{x})-U(\mathrm{x})|=U_{\mu}(\mathrm{x})-U(\mathrm{x})\leq\frac{L\mu^{1+\alpha}d^{\frac{1+\alpha}{p}}}{1+\alpha}}$.

(ii) $\forall\mathrm{x},\ \mathrm{y}\in\mathbb{R}^{d}$: ${\displaystyle \Vert\nabla U_{\mu}(\mathrm{y})-\nabla U_{\mu}(\mathrm{x})\Vert_{2}\leq\frac{Ld^{\frac{1-\alpha}{p}}}{\mu^{1-\alpha}(1+\alpha)^{1-\alpha}}\Vert\mathrm{y}-\mathrm{x}\Vert_{2}}$.\end{lem}
\begin{proof}
See Appendix B.1
\end{proof}
\section{Black-box sampling for regularized potential}

\label{sec:AIF}

To study the convergence of the continuous-time process (involving
strong convexity), we work with regularized potentials that have the
following composite structure: 
\begin{equation}
\overline{U}(\mathrm{x}):=U(\mathrm{x})+\frac{\lambda}{2}\left\Vert x\right\Vert ^{2},\label{eq:2.5b}
\end{equation}
where $\lambda>0$. Observe that by the triangle inequality: 
\begin{align}
\Vert\nabla\overline{U}(\mathrm{x})-\nabla\overline{U}(\mathrm{y})\Vert_{2} & \leq\Vert\nabla U(\mathrm{x})-\nabla U(\mathrm{y})\Vert_{2}+\lambda\Vert\mathrm{x}-\mathrm{y}\Vert_{2}\label{eq:2.6}\\
~ & \leq L\Vert\mathrm{x}-\mathrm{y}\Vert_{2}^{\alpha}+\lambda\Vert\mathrm{x}-\mathrm{y}\Vert_{2}.\nonumber 
\end{align}
From which, by employing Lemma \ref{2.1}, we get: 
\begin{equation}
{\displaystyle \Vert\nabla\overline{U}_{\mu}(\mathrm{x})-\nabla\overline{U}_{\mu}(\mathrm{y})\Vert_{2}\leq\left(\frac{Ld^{\frac{1-\alpha}{p}}}{\mu^{1-\alpha}(1+\alpha)^{1-\alpha}}+\lambda\right)\Vert\mathrm{x}-\mathrm{y}\Vert_{2}.}\label{eq:2.8}
\end{equation}
In this case, $U(\cdot)$ is $(L,\alpha)$-weakly smooth (possibly
with $\alpha=0$, $U$ is non-smooth and Lipschitz-continuous). We
now analyze LMC where we perturb the points at which gradients of
$\overline{U}$ are evaluated by a $p$-generalized Gaussian random
variable. Note that it remains unclear whether it is possible to achieve
such bounds for (LMC) without this perturbation. Recall that LMC in
terms of the potential $\overline{U}$ can be specified as: 
\begin{equation}
\mathrm{x}_{k+1}=\mathrm{x}_{k}-\eta\nabla\overline{U}_{\mu}(\mathrm{x}_{k})+\sqrt{2\eta}\varsigma_{k},\label{eq:LMC}
\end{equation}
where $\varsigma_{k}\sim N(0,\ I_{d\times d})$ are independent Gaussian
random vectors. This method is actually the Euler-Maruyama discretization
of the Langevin diffusion. Rather than working with the algorithm specified by Eq. \ref{eq:LMC},
we define an estimate $g_{\mu,n}(\theta)$, of the gradient $\nabla\overline{U}_{\mu}(\theta)$
as follows: 
\begin{align}
g_{\mu,n}(x)=\frac{1}{n}\sum_{i=1}^{n}\frac{\overline{U}_{\mu}(x+\mu\xi_{i})-\overline{U}_{\mu}(x)}{\mu}\left(\xi_{i}\circ\left|\xi_{i}\right|^{p-2}\right)\label{eq:gradest}
\end{align}
where $\xi_{i}\sim N_{p}(0,I_{d})$ and are i.i.d. and $\circ$ stands
for the Hadamard product and $\left|\cdot\right|$ is used for absolute
value of each component of the vector $\xi_{i}$. The gradient estimator
in Equation~\ref{eq:gradest} is a consequence of a p-generalized
Gaussian identity we derive in the proof of Lemma \ref{2.1}. Based on the above
estimate of the gradient, we obtain the following result.

\begin{lem} \label{3.1-1} For any $\mathrm{x_{k}}\in\mathbb{R}^{d}$,
then $g_{\mu,n}(x_{k})=\nabla\overline{U}_{\mu}(x_{k})+\zeta_{k}$ is an
estimator of $\nabla\overline{U}_{\mu}$ such that
\begin{align*}
\|E[\zeta_{k}\mid x_{k}]\|^{2} & \leq\left(M+\lambda\right)^{2}\mu^{2}d^{\frac{2}{p}}\\
~E\left[\left\Vert \zeta_{k}-E[\zeta_{k}\mid x_{k}]\right\Vert ^{2}\right] & \leq\frac{1}{n}\left(\frac{1}{2}\left(M+\lambda\right)\mu\left(d+3\right){}^{3\slash p}+\sqrt{2}\left(d+2\right)^{\frac{2}{p}}\left\Vert \nabla\overline{U}_{\mu}(x_{k})\right\Vert \right)^{2}
\end{align*}
\end{lem} 

\begin{proof} See Appendix B.2 \end{proof} 

Let the distribution of the $k^{th}$ iterate $\mathrm{y}_{k}$ be
represented by $\overline{\pi}_{\mu,k}$, and let $\overline{\pi}_{\mu}\propto\exp(-\overline{U}_{\mu})$
be the distribution with $\overline{U}_{\mu}$ as the potential. It is straightfoward that $p$-generalized Gaussian smoothing of regularized potential remains strong convexity so we omit the proof. 
Our overall tactics for showing our core results are as follows. First,
we prove that the $p$-generalized Gaussian smoothing does not alter
the objective distribution substantially in term of the Wasserstein
distance, by bounding $W_{2}(\overline{\pi},\overline{\pi}_{\mu})$
(Lemma \ref{3.2}). We then deploy a result
on mixing times of Langevin diffusion with stochastic gradients, which
enables us to bound $W_{2}(\overline{\pi}_{\mu, k},\overline{\pi}_{\mu})$.
\begin{lem} \label{3.2}Assume that $\overline{\pi}\propto\exp(-\overline{U})$
and $\overline{\pi}_{\mu}\propto\exp(-\overline{U}_{\mu})$. We deduce
the following bounds, for any $\lambda\geq0$ 
\[
W_{2}^{2}(\overline{\pi},\ \overline{\pi}_{\mu})\leq\frac{4(d+\lambda\Vert x^{*}\Vert_{2}^{2})}{\lambda}\left(a+e^{a}-1\right).
\]
where $x^{*}$ is a unique minimizer of $\overline{U}$, $a=\frac{L\mu^{1+\alpha}d^{\frac{1+\alpha}{p}}}{1+\alpha}+\frac{1}{2}\lambda\mu^{2}\left(d+1\right)^{\frac{2}{p}}$.
If, in addition, $a\leq0.1$ and for sufficient small $\lambda$,
then 
\[
W_{2}(\overline{\pi},\ \overline{\pi}_{\mu})\leq3\sqrt{\frac{da}{\lambda}},
\]
\end{lem} \begin{proof} See Appendix B.3 \end{proof}
Our primary
outcome is reported in the subsequent theorem. 
\begin{thm} \label{Theorem3.2}Let
the initial iterate $\mathrm{y}_{0}$ be drawn from a probability
distribution $\overline{\pi}_{\mu, 0}$. If the step size $\eta$ satisfies
$\eta<2/(M+2\lambda)$ and for sufficiently small $\lambda$, then:
\begin{align*}
& W_{2}(\overline{\pi}_{\mu, K} ,{\pi})  \leq(1-0.5\lambda\eta)^{K/2}W_{2}(\overline{\pi}_{\mu, 0},\overline{\pi}_{\mu})+1.9\frac{\left(M+\lambda\right)}{\lambda}(\eta d)^{1\slash2}+2\frac{\left(M+\lambda\right)}{\lambda}\mu d^{1\slash p}+3\sqrt{\frac{da}{\lambda}}\\
 & +\frac{1}{\sqrt{\lambda}}\cdot\frac{1}{\sqrt{n}}\eta^{1\slash2}\left(M+\lambda\right)\mu(d+3)^{3\slash p}+\frac{\sqrt{\left(M+\lambda\right)}}{\sqrt{\lambda}}\cdot\frac{1}{\sqrt{n}}\eta^{1\slash2}d^{1/2}(d+2)^{1/p}+C\lambda\left(dlogd\right)^{4},
\end{align*}
where ${\displaystyle M=\frac{Ld^{\frac{1-\alpha}{p}}}{\mu^{1-\alpha}(1+\alpha)^{1-\alpha}}}$,
${\displaystyle a=\frac{L\mu^{1+\alpha}d^{\frac{1+\alpha}{p}}}{1+\alpha}+\frac{1}{2}\mu^{2}\lambda(d+1)^{\frac{2}{p}}}$\end{thm} 

\begin{proof} Adapting the technique of analysis of \citep{dalalyan2019user,shen2019non}, by a triangle inequality, the Wasserstein distance between $\overline{\pi}_{\mu, K}$
and $\pi$ is bounded by: 
\begin{center}
$W_{2}(\overline{\pi}_{\mu, K},\pi)\leq W_{2}(\overline{\pi}_{\mu, K},\,\overline{\pi}_{\mu})+W_{2}(\overline{\pi}_{\mu},\overline{\pi})+W_{2}(\overline{\pi},\,\pi)$. 
\par\end{center}

We use \citep{dalalyan2019user} Theorem 9 to bound the first term
$W_{2}(\overline{\pi}_{K},\,\overline{\pi}_{\mu})$. Recall that $\overline{U}_{\mu}$
is continuously differentiable, $(M+\lambda)$-smooth (with $M=\frac{Ld^{(\frac{1-\alpha}{p})}}{\mu^{1-\alpha}(1+\alpha)^{1-\alpha}}$)
and $\lambda$-strongly convex. Additionally, $\{\mathrm{y}_{k}\}_{k=1}^{K}$
can be viewed as iterates of over-damped Langevin MCMC with respect
to the potential specified by $\overline{U}_{\mu}$ and is updated
using biased noisy gradients of $\overline{U}_{\mu}$. For any random vector $\boldsymbol{X}$,
we define the norm $\|\boldsymbol{X}\|_{L_{2}}=(E[\|\boldsymbol{X}\|_{2}^{2}])^{1/2}$.
Let $L_{0}$ be a random vector drawn from $\overline{\pi}_{\mu}$ such that $W_{2}(\nu_{k},\overline{\pi}_{\mu})=\|L_{0}-x_{k}\|_{L_{2}}$ and $E[\zeta_{k}|x_{k},L_{0}]=E[\zeta_{k}|x_{k}]$. Let $W$ be a d-dimensional Brownian motion independent of $(x_{k},L_{0},\zeta_{k})$, such that $W_{\eta}=\sqrt{\eta}\,\zeta_{k+1}$. We define the stochastic process $L$ so that 
$L_{t}=L_{0}-\int_{0}^{t}\nabla \overline{U}_{\mu}(L_{s})\,ds+\sqrt{2}\,W_{t},\ \forall\,t>0$. It is clear that this equation implies that ${\displaystyle L_{\eta}}=L_{0}-\int_{0}^{\eta}\nabla \overline{U}_{\mu}(L_{s})\,ds+\sqrt{2}\,W_{\eta}=L_{0}-\int_{0}^{\eta}\nabla \overline{U}_{\mu}(L_{s})\,ds+\sqrt{2\eta}\,{\zeta}_{k+1}$. Furthermore, $\{L_{t}:t\ge0\}$ is a diffusion process having $\overline{\pi}_{\mu}$ as the stationary distribution. Since the initial value $L_{0}$ is drawn from $\overline{\pi}_{\mu}$, we have $L_{t}\sim\overline{\pi}_{\mu}$ for every $t\ge0$.

Let us define $\Delta_{k}=L_{0}-x_{k},\;\Delta_{k+1}=L_{\eta}-x_{k+1}$ and $V=\int_{0}^{\eta}[\nabla\overline{U}_{\mu}(L_{s})-\nabla\overline{U}_{\mu}(L_{0})]ds$ then $\|V\|_{L_{2}}\leq\frac{1}{2}\big(\eta^{4}(\lambda+M)^{3}d\big)^{1/2}+\frac{2}{3}(2\eta^{3}d)^{1/2}(\lambda+M)\leq.95(\lambda+M)\eta^{\frac{3}{2}}d^{\frac{1}{2}}$ for $\eta$ small enough by Lemma 4 in \citep{dalalyan2019user}. Thus, 
\[\|\Delta_{k+1}\|_{L_{2}}  =\|\Delta_{k}-\eta[\nabla\overline{U}_{\mu}(x_{k}+\Delta_{k})-\nabla\overline{U}_{\mu}(x_{k})]-V+\eta\zeta_{k}\|_{L_{2}}\]
\[\leq\{\|\Delta_{k}-\eta\lambda\Delta_{k}\|_{L_{2}}^{2}+\eta^{2}\|\zeta_{k}-E[\zeta_{k}\mid x_{k}]\|_{L_{2}}^{2}\}^{1\slash2}+\|V\|_{L_{2}}+\eta\|E[\zeta_{k}\mid x_{k}]\|_{L_{2}}\]
\[\leq\left\{ (1-\lambda\eta)^{2}\|\Delta_{k}\|_{L_{2}}^{2}+\frac{\eta^{2}}{n}\left(\frac{1}{2}\left(M+\lambda\right)\mu(d+3)^{3\slash p}+(\|\nabla\overline{U}_{\mu}(L_{0})\|_{L_{2}}+\left(M+\lambda\right)\|\Delta_{k}\|_{L_{2}})\sqrt{2}\left(d+2\right)^{1/p}\right)^{2}\right\} ^{1\slash2}\]
\[\qquad+.95\left(M+\lambda\right)(\eta^{3}d)^{1\slash2}+\left(M+\lambda\right)\mu\eta d{}^{1\slash p}\]
\[\leq\left\{ (1-\lambda\eta)^{2}\|\Delta_{k}\|_{L_{2}}^{2}+\frac{2\eta^{2}}{n}\left(\frac{1}{2}\left(M+\lambda\right)\mu(d+3)^{3\slash p}+\|\nabla\overline{U}_{\mu}(L_{0})\|_{L_{2}}\sqrt{2}\left(d+2\right)^{1/p}\right)^{2}\right\} ^{1\slash2}\]
 \[+\frac{2\left(M+\lambda\right)^{2}\eta^{2}\left(d+2\right)^{\frac{2}{p}}}{n(1-\lambda\eta)}\|\Delta_{k}\|_{L_{2}}+.95\left(M+\lambda\right)(\eta^{3}d)^{1\slash2}+\left(M+\lambda\right)\mu\eta d^{1\slash p}\]
\[\leq\left\{ (1-\lambda\eta)^{2}\|\Delta_{k}\|_{L_{2}}^{2}+\frac{2\eta^{2}}{n}\left(\frac{1}{2}\left(M+\lambda\right)\mu(d+3)^{3\slash p}+\sqrt{2\left(M+\lambda\right)d}\left(d+2\right)^{1/p}\right)^{2}\right\} ^{1\slash2}\]
\[+\frac{\sqrt{2}\left(M+\lambda\right)^{2}\eta^{2}\left(d+2\right)^{2/p}}{n(1-\lambda\eta)}\|\Delta_{k}\|_{L_{2}}+.95\left(M+\lambda\right)(\eta^{3}d)^{1\slash2}+\left(M+\lambda\right)\mu\eta d^{1\slash p}.
\]
Here we use the fact that $\sqrt{a^{2}+b+c}\leq\sqrt{a^{2}+b}+\frac{c}{2a}$
and $E[\|\nabla\overline{U}_{\mu}(X_{s})\|^{2}]\leq\left(M+\lambda\right)d$.
By Lemma 9 in \citep{dalalyan2019user}, for $n$ large enough, $\frac{\sqrt{2}\left(M+\lambda\right)^{2}\eta^{2}\left(d+2\right)^{2/p}}{n(1-\lambda\eta)}\leq0.5\lambda\eta$,
the above inequality leads to 
\begin{align*}
\|\Delta_{k}\|_{L_{2}} & \leq\left(1-0.5\lambda\eta\right)^{k}\|\Delta_{0}\|_{L_{2}}+1.9\frac{\left(M+\lambda\right)}{\lambda}(\eta d)^{1/2}+2\frac{\left(M+\lambda\right)}{\lambda}\mu d^{1\slash p}\\
 & \qquad\frac{1}{\sqrt{\lambda}}\cdot\frac{1}{\sqrt{n}}\eta^{1\slash2}\left(M+\lambda\right)\mu(d+3)^{3\slash p}+\frac{\sqrt{\left(M+\lambda\right)}}{\sqrt{\lambda}}\cdot\frac{1}{\sqrt{n}}\eta^{1\slash2}d^{1/2}(d+2)^{1/p}.
\end{align*}
Therefore, we obtain the bound in Wasserstein distance. 
\begin{align*}
W_{2}(\overline{\pi}_{\mu, k},\overline{\pi}_{\mu}) & \leq\left(1-0.5\lambda\eta\right)^{k}W_{2}(\overline{\pi}_{\mu, 0},\overline{\pi}_{\mu})+1.9\frac{\left(M+\lambda\right)}{\lambda}(\eta d)^{1\slash2}+2\frac{\left(M+\lambda\right)}{\lambda}\mu d^{1\slash p}\\
 & \qquad\frac{1}{\sqrt{\lambda}}\cdot\frac{1}{\sqrt{n}}\eta^{1\slash2}\left(M+\lambda\right)\mu(d+3)^{3\slash p}+\frac{\sqrt{\left(M+\lambda\right)}}{\sqrt{\lambda}}\cdot\frac{1}{\sqrt{n}}\eta^{1\slash2}d^{1/2}(d+2)^{1/p}.
\end{align*}
By \citep{dalalyan2019bounding} Propositions 1 and 2, 
\[
W_{2}^{2}(\overline{\pi},\pi)\le111\lambda m_{2}(\pi)^{2},
\]
 where $m_{2}(\pi):=\int_{\mathbb{R}^{d}}\|x-x^{*}\|_{2}^{2}\,\pi(x)\,dx=C\left(dlogd\right)^{2}$
for some constant $C$ independent of $d.$ 
Combining with Lemma \ref{3.2} produces the desired result. \end{proof}
\begin{rem} 

When taking the limit as $\mu$ goes to $0$, our gradient estimate is unbiased, our bound is equivalent to $O(d^{\frac{2+p}{2p}}+\frac{d^{\frac{1+\alpha+p}{2p}}}{\sqrt{\lambda}})$ for non-regularization potential.
In comparison to the previous result of \citep{chatterji2019langevin},
which is of order $O(\frac{d}{\lambda})$ when considered Lipschitz
constant $L$ as a constant and $p=2$, our approach has better dependencies
on strongly convex parameter $\lambda$. In general, our result is
strictly better whenever $p<2$. In addition, our method is more generalized
than Gaussian smoothing, especially when we desire to explore the
space by heavier tail distribution than normal. \end{rem}

\section{Conclusion\label{sec:5Conclusion}}
We derive polynomial-time theoretical assurances for an LMC
that uses $p$-generalized Gaussian smoothing. 
The algorithm we proposed is a generalization of the recent Gaussian smoothing LMC method.
It is potential to broaden our results to sampling from distributions with non-smooth
and non-convex structure or integrate into derivative-free LMC algorithm.
In addition, we speculate that the dependence on $d$ and $\lambda$ is not optimal and can be improved to match those obtained for the 2-Wasserstein distance using proximal assumption.

\appendix
\setcounter{lem}{3}
\section{Distance measures}
Define a transference plan $\zeta$, a distribution on $(\mathbb{R}^{d}\times\mathbb{R}^{d},\ \mathcal{B}(\mathbb{R}^{d}\times\mathbb{R}^{d}))$ (where $\mathcal{B}(\mathbb{R}^{d}\times\mathbb{R}^{d})$ is the Borel $\sigma$-field
of ($\mathbb{R}^{d}\times\mathbb{R}^{d}$))
so that $\zeta(A\times\mathbb{R}^{d})=P(A)$ and $\zeta(\mathbb{R}^{d}\times A)=Q(A)$ for any $A\in \mathcal{B}(\mathbb{R}^{d})$. Let $\Gamma(P,\ Q)$ designate the set of all such transference plans. Then the 2-Wasserstein distance is formulated as: 
\[
W_{2}(P,Q)\ :=\left(\inf_{\zeta\in\Gamma(P,Q)}\int_{\mathrm{x},\mathrm{y}\in\mathbb{R}^{d}}\Vert\mathrm{x}-\mathrm{y}\Vert_{2}^{2}\mathrm{d}\zeta(\mathrm{x},\ \mathrm{y})\right)^{1/2}.
\]
\section{Proofs}
\subsection{Proof of Lemma 1}
\begin{proof}
(i). Since $U_{\mu}(\mathrm{x})=\mathrm{E}_{\xi}[U(\mathrm{x}+\mu\xi)]$,
$U(\mathrm{x})=\mathrm{E}_{\xi}[U(x)]$ and $E_{\xi}(\xi)=0$, which
implies $E_{\xi}\mu\left\langle \nabla U(\mathrm{x}),\ \xi\right\rangle =0$,
we obtain
$$
U_{\mu}(\mathrm{x})-U(\mathrm{x})=E\left[U(\mathrm{x}+\mu\xi)-U(\mathrm{x})-\mu\left\langle \nabla U(\mathrm{x}),\ \xi\right\rangle \right].
$$
Because $U$ is convex and $\mu>0$, from A.1, $U(\mathrm{x}+\mu\xi)-U(\mathrm{x})-\mu\left\langle \nabla U(\mathrm{x}),\ \xi\right\rangle \geq0$
for every $\xi$ and $x$, so $U_{\mu}(\mathrm{x})\geq U(\mathrm{x}),\ \forall\mathrm{x}$.
By the definition of the density of $p$-generalized Gaussian distribution, we also have: 
$$
U_{\mu}(\mathrm{x})-U(\mathrm{x})=\frac{1}{\kappa}\int_{\mathbb{R}^{d}}[U(\mathrm{x}+\mu\xi)-U(\mathrm{x})-\mu\left\langle \nabla U(\mathrm{x}),\ \xi\right\rangle ]e^{-\left\Vert \xi\right\Vert _{p}^{p}/p}\mathrm{d}\xi.
$$
Applying Eq. \ref{eq:4} and previous inequality: 
\begin{align*}
|U_{\mu}(\mathrm{x})-U(\mathrm{x})| & \leq\frac{L}{\kappa(1+\alpha)}\mu^{1+\alpha}\int_{\mathbb{R}^{d}}\left\Vert \xi\right\Vert _{2}^{(1+\alpha)}e^{-\left\Vert \xi\right\Vert _{p}^{p}/p}\mathrm{d}\xi\\
& \leq\frac{L}{\kappa(1+\alpha)}\mu^{1+\alpha}\int_{\mathbb{R}^{d}}\left\Vert \xi\right\Vert _{p}^{(1+\alpha)}e^{-\left\Vert \xi\right\Vert _{p}^{p}/p}\mathrm{d}\xi\\
& =\frac{L}{\kappa(1+\alpha)}\mu^{1+\alpha}\int_{\mathbb{R}^{d}}\left\Vert \xi\right\Vert _{p}^{2p(1+\alpha)/(2p)}e^{-\left\Vert \xi\right\Vert _{p}^{p}/p}\mathrm{d}\xi.
\end{align*}
The second inequality follows from $\left\Vert \xi\right\Vert _{p}\geq\left\Vert \xi\right\Vert _{2}$ for any $p\leq2$. Since $1\leq p$ we have $1+\alpha\leq2p$ and $x^{\frac{1+\alpha}{2p}}$ is concave, it follows that 
\begin{eqnarray*}
U_{\mu}(\mathrm{x})-U(\mathrm{x}) & \leq\frac{L}{1+\alpha}\mu^{1+\alpha}E\left[\left\Vert \xi\right\Vert _{p}^{\frac{2p(1+\alpha)}{2p}}\right] & \leq\frac{L}{1+\alpha}\mu^{1+\alpha}\left[E\left\Vert \xi\right\Vert _{p}^{2p}\right]^{\frac{1+\alpha}{2p}}\\
\leq\frac{L}{1+\alpha}\mu^{1+\alpha}\left[\frac{2p\Gamma\left(\frac{d}{p}+2\right)}{\Gamma\left(\frac{d}{p}\right)}\right]^{\frac{1+\alpha}{2p}} & =\frac{L}{1+\alpha}\mu^{1+\alpha}\left(\frac{2pd\left(d+p\right)}{p^{2}}\right)^{\frac{1+\alpha}{2p}} & =\frac{L}{1+\alpha}\mu^{1+\alpha}\left(\frac{2d\left(d+p\right)}{p}\right)^{\frac{1+\alpha}{2p}}.
\end{eqnarray*}
Finally, for sufficiently large $d$, we have 
\[
|U_{\mu}(\mathrm{x})-U(\mathrm{x})|\leq\frac{L}{1+\alpha}\mu^{1+\alpha}d^{\frac{1+\alpha}{p}}.
\]
(ii). We adapt the technique of \citep{nesterov2017random} to $p$-generalized
Gaussian smoothing. Let $y=\mathrm{x}+\mu\xi$, then $U_{\mu}(\mathrm{x})$
is rewritten in another form as 
\[
U_{\mu}(\mathrm{x})=\frac{1}{\mu^{d}\kappa}\int_{\mathbb{R}^{d}}U(y)e^{-\frac{1}{p\mu^{p}}\left\Vert y-x\right\Vert _{p}^{p}}\mathrm{d}y.
\]
Now taking the gradient with respect to $x$ of $U_{\mu}(\mathrm{x})$
gives 
\[
\nabla_{x}U_{\mu}(\mathrm{x})=\frac{1}{\mu^{d}\kappa}\nabla_{x}\int_{\mathbb{R}^{d}}U(y)e^{-\frac{1}{p\mu^{p}}\left\Vert y-x\right\Vert _{p}^{p}}\mathrm{d}y.
\]
By Fubini Theorem with some regularity (i.e. $E|U(\mathrm{y})|<\infty$), we can exchange the gradient
and integral and get 
\begin{align*}
\nabla_{x}U_{\mu}(\mathrm{x}) & =\frac{1}{\mu^{d}\kappa}\int_{\mathbb{R}^{d}}\nabla_{x}\left(U(y)e^{-\frac{1}{p\mu^{p}}\left\Vert y-x\right\Vert _{p}^{p}}\right)\mathrm{d}y\\
 & =\frac{1}{\mu^{d}\kappa}\int_{\mathbb{R}^{d}}U(y)\nabla_{x}\left(e^{-\frac{1}{p\mu^{p}}\left\Vert y-x\right\Vert _{p}^{p}}\right)\mathrm{d}y\\
 & =\frac{1}{\mu^{d}\kappa}\int_{\mathbb{R}^{d}}U(y)e^{-\frac{1}{p\mu^{p}}\left\Vert y-x\right\Vert _{p}^{p}}\frac{1}{\mu^{p}}\mathrm{\left\Vert y-x\right\Vert _{p}^{p-1}\nabla_{x}\left(\left\Vert y-x\right\Vert _{p}\right)d}y.\\
 & =-\frac{1}{\mu^{d}\kappa}\int_{\mathbb{R}^{d}}U(y)e^{-\frac{1}{p\mu^{p}}\left\Vert y-x\right\Vert _{p}^{p}}(y-x)\circ\left|y-x\right|^{p-2}dy.
\end{align*}
where $\circ$ stands for the Hadamard product and $\left|\cdot\right|$
is used for absolute value of each component of the vector $y-x$.
Thereby, by changing variable back to $\xi$, we deduce 
\begin{align*}
\nabla_{x}U_{\mu}(\mathrm{x}) & =-\frac{1}{\mu\kappa}\int_{\mathbb{R}^{d}}U(\mathrm{x}+\mu\xi)e^{-\frac{1}{p}\left\Vert \xi\right\Vert _{p}^{p}}\xi\circ\left|\xi\right|^{p-2}\mathrm{d}\xi.\\
 & =-E_{\xi}\left[U(\mathrm{x}+\mu\xi)\xi\circ\left|\xi\right|^{p-2}\right]
\end{align*}
Thus, by Jensen’s inequality:
\begin{equation}
{\displaystyle \Vert\nabla U_{\mu}(\mathrm{y})-\nabla U_{\mu}(\mathrm{x})\Vert_{2}} \leq\frac{1}{\mu\kappa}\int_{\mathbb{R}^{d}}\left\Vert U(\mathrm{y}+\mu\xi)- U(\mathrm{x}+\mu\xi)\right\Vert_{2}\left\Vert \xi\circ\left|\xi\right|^{p-2}\right\Vert_{2}e^{-\left\Vert \xi\right\Vert _{p}^{p}/p}\mathrm{d}\xi.\label{eq:2.3}
\end{equation}
Further, by using generalized Holder inequality for $1\leq p\leq2$, $\left\Vert \xi\circ\left|\xi\right|^{p-2}\right\Vert _{2}$ can be bounded as: 
\[\left\Vert \xi\circ\left|\xi\right|^{p-2}\right\Vert _{2} =\left\Vert \xi^{p-1}\right\Vert _{2}\\
\leq\left\Vert \xi^{p-1}\right\Vert _{p}\\
=\left\Vert \xi^{p-1}\cdot1_{d}\right\Vert _{p}\\
\leq\left\Vert \xi\right\Vert _{p}^{p-1}\left\Vert 1_{d}\right\Vert _{p}^{2-p}\\
=\left\Vert \xi\right\Vert _{p}^{p-1}d^{\frac{2-p}{p}}.
\]
Thus, applying Jensen’s inequality, we derive: 
\begin{equation}
{\displaystyle \Vert\nabla U_{\mu}(\mathrm{y})-\nabla U_{\mu}(\mathrm{x})\Vert_{2}\leq\frac{1}{\mu}d^{\frac{2-p}{p}}\left(\frac{p^{1-\frac{1}{p}}}{2\Gamma(\frac{1}{p})}\right)^{d}\int_{\mathbb{R}^{d}}|U(\mathrm{x}+\mu\xi)-U(\mathrm{y}+\mu\xi)|\left\Vert \xi\right\Vert _{p}^{p-1}e^{-\left\Vert \xi\right\Vert _{p}^{p}/p}\mathrm{d}\xi.}\label{eq:2.4}
\end{equation}
Using Eq. \ref{eq:4}, we obtain: 
\begin{align*}
|U(\mathrm{x}+\mu\xi)-U(\mathrm{y}+\mu\xi)| & \leq\min\left\{ \left\langle \nabla U(\mathrm{y}+\mu\xi),\mathrm{x}-\mathrm{y}\right\rangle +\frac{L}{1+\alpha}\Vert\mathrm{y}-\mathrm{x}\Vert_{2}^{1+\alpha},\right.\\
& \left.\left\langle \nabla U(\mathrm{x}+\mu\xi),\mathrm{y}-\mathrm{x}\right\rangle +\frac{L}{1+\alpha}\Vert\mathrm{y}-\mathrm{x}\Vert_{2}^{1+\alpha}\right\} \\
& \leq\frac{1}{2}\left\langle \nabla U(\mathrm{y}+\mu\xi)-\nabla U(\mathrm{x}+\mu\xi),\ \mathrm{x}-\mathrm{y}\right\rangle +\frac{L}{1+\alpha}\Vert\mathrm{y}-\mathrm{x}\Vert_{2}^{1+\alpha}\\
& \leq\frac{L}{1+\alpha}\Vert\mathrm{y}-\mathrm{x}\Vert_{2}^{1+\alpha},
\end{align*}
where the second inequality comes from the minimum being smaller than
the mean, and the last inequality is achieved by convexity of $U$ (which implies
$\left\langle \nabla U(\mathrm{x})-\nabla U(\mathrm{y}),x-y\right\rangle \geq0$,
$\forall\mathrm{x}$, $y$). Thus, combining with Eq. \ref{eq:2.4},
we induce: 
\begin{align}
\Vert\nabla U_{\mu}(\mathrm{y})-\nabla U_{\mu}(\mathrm{x})\Vert_{2} & \leq\frac{L}{\mu(1+\alpha)}d^{\frac{2-p}{p}}\Vert\mathrm{y}-\mathrm{x}\Vert_{2}^{1+\alpha}\left(\frac{p^{1-\frac{1}{p}}}{2\Gamma(\frac{1}{p})}\right)^{d}\int_{\mathbb{R}^{d}}\left\Vert \xi\right\Vert _{p}^{p-1}e^{-\left\Vert \xi\right\Vert _{p}^{p}/p}\mathrm{d}\xi\nonumber \\
& =\frac{L}{\mu(1+\alpha)}d^{\frac{2-p}{p}}\Vert\mathrm{y}-\mathrm{x}\Vert_{2}^{1+\alpha}\left[\frac{p^{\frac{p-1}{p}}\Gamma(\frac{n+p-1}{p})}{\Gamma(\frac{n}{p})}\right]\\
& \leq\frac{L}{\mu(1+\alpha)}\Vert\mathrm{y}-\mathrm{x}\Vert_{2}^{1+\alpha}d^{\frac{1}{p}}.\label{eq:2.5}
\end{align}
Finally, combining A.2 and \eqref{eq:2.5}: 
\[
\Vert\nabla U_{\mu}(\mathrm{y})-\nabla U_{\mu}(\mathrm{x})\Vert_{2}=\Vert\nabla U_{\mu}(\mathrm{y})-\nabla U_{\mu}(\mathrm{x})\Vert_{2}^{\alpha}\cdot\Vert\nabla U_{\mu}(\mathrm{y})-\nabla U_{\mu}(\mathrm{x})\Vert_{2}^{1-\alpha}
\]
\[
\leq\frac{Ld^{(\frac{1-\alpha}{p})}}{\mu^{1-\alpha}(1+\alpha)^{1-\alpha}}\Vert\mathrm{y}-\mathrm{x}\Vert_{2},
\]
as claimed.
\end{proof}
\subsection{Proof of Lemma 2}
\begin{proof} Under smooth assumption, 
\begin{align*}
\|E[\zeta_{k}\mid x_{k}]\|^{2} & =\|E[g_{\mu,n}(x_{k})-\nabla\overline{U}_{\mu}(x_{k})\mid x_{k}]\|^{2}\\
 & \leq\|E[g_{\mu,1}(x_{k})-\nabla\overline{U}_{\mu}(x_{k})\mid x_{k}]\|^{2}\\
 & =\left\Vert E\left[\nabla\overline{U}_{\mu}(x_{k}+\mu\xi)|x_{k}\right]-\nabla\overline{U}_{\mu}(x_{k})\right\Vert ^{2}\\
 & \leq E\left[\left(M+\lambda\right)\mu\left\Vert \xi\right\Vert \right]^{2}\\
 & =\left(M+\lambda\right)^{2}\mu^{2}E^{2}\left(\left\Vert \xi\right\Vert \right)\\
 & \stackrel{_{1}}{\leq}\left(M+\lambda\right)^{2}\mu^{2}d^{\frac{2}{p}},
\end{align*}

where $1$ follows from Lemma \ref{L4} below.

For the variance, by independence of p-generalized Gaussian sample
$\xi_{i}$'s, we have the following observation. 
\begin{align*}
E\left[\left\Vert \zeta_{k}-E[\zeta_{k}\mid x_{k}]\right\Vert ^{2}\right] & =E\left[\left\Vert g_{\mu,n}(x_{k})-\overline{U}_{\mu}(x_{k})\right\Vert ^{2}\right]\\
 & =\frac{1}{n}E\left[\left\Vert g_{\mu,1}(x_{k})-\overline{U}_{\mu}(x_{k})\right\Vert ^{2}\right]
\end{align*}

In order to obtain the variance of $g_{\mu,1}(x_{k})$, we split the
centered error term into three parts. 
\begin{align*}
\zeta_{k}-E[\zeta_{k}\mid x_{k}] & =\frac{\overline{U}_{\mu}(x_{k}+\mu\xi)-\overline{U}_{\mu}(x_{k})}{\mu}\left(\xi\circ\left|\xi\right|^{p-2}\right)-\nabla\overline{U}_{\mu}(x_{k})\\
 & =\frac{\overline{U}_{\mu}(x_{k}+\mu\xi)-\overline{U}_{\mu}(x_{k})-\left\langle \nabla\overline{U}_{\mu}(x_{k}),\mu\xi\right\rangle }{\mu}\left(\xi\circ\left|\xi\right|^{p-2}\right)+\left\langle \nabla\overline{U}_{\mu}(x_{k}),\xi\right\rangle \left(\xi\circ\left|\xi\right|^{p-2}\right)-\nabla\overline{U}_{\mu}(x_{k})\\
 & \stackrel{_{def}}{=}A+B-C.
\end{align*}
We have 
\begin{align*}
E[\|A\|^{2}\mid x_{k}] & =E\left[\left.\left\Vert \frac{\overline{U}_{\mu}(x_{k}+\mu\xi)-\overline{U}_{\mu}(x_{k})-\left\langle \nabla\overline{U}_{\mu}(x_{k}),\mu\xi\right\rangle }{\mu}\xi\circ\left|\xi\right|^{p-2}\right\Vert ^{2}\right|x_{k}\right]\\
 & \leq E\left[\left(\tfrac{1}{2}\left(M+\lambda\right)\mu\left\Vert \xi\right\Vert ^{2}\right)^{2}\left\Vert \xi\circ\left|\xi\right|^{p-2}\right\Vert ^{2}\mid x_{k}\right]\\
 & \leq E\left[\tfrac{1}{4}\left(M+\lambda\right)^{2}\mu^{2}\left\Vert \xi\right\Vert ^{4}\left\Vert \xi\right\Vert _{p}^{2p-2}d^{\frac{4-2p}{p}}\mid x_{k}\right]\\
 & \leq E\left[\tfrac{1}{4}\left(M+\lambda\right)^{2}\mu^{2}d^{\frac{4-2p}{p}}\left\Vert \xi\right\Vert _{p}^{2p+2}\mid x_{k}\right]\\
 & \leq\tfrac{1}{4}M\left(M+\lambda\right)^{2}\mu^{2}d^{\frac{4-2p}{p}}E[\left\Vert \xi\right\Vert _{p}^{2p+2}]\\
 & \leq\frac{1}{4}\left(M+\lambda\right)^{2}\mu^{2}d^{\frac{4-2p}{p}}(d+p+1)^{\frac{2p+2}{p}},\\
 & \leq\frac{1}{4}\left(M+\lambda\right)^{2}\mu^{2}(d+p+1)^{\frac{6}{p}},\\
E[\|B\|^{2}\mid x_{k}] & =E[\|\left\langle \nabla\overline{U}_{\mu}(x_{k}),\xi\right\rangle \left(\xi\circ\left|\xi\right|^{p-2}\right)\|^{2}\mid x_{k}]\\
 & \leq E[\left\Vert \nabla\overline{U}_{\mu}(x_{k})\right\Vert ^{2}\left\Vert \xi\right\Vert ^{2}\left\Vert \xi\circ\left|\xi\right|^{p-2}\right\Vert ^{2}\mid x_{k}]\\
 & \leq E[\left\Vert \nabla\overline{U}_{\mu}(x_{k})\right\Vert ^{2}\left\Vert \xi\right\Vert ^{2}\left\Vert \xi\right\Vert _{p}^{2p-2}d^{\frac{4-2p}{p}}\mid x_{k}]\\
 & \leq E\left[d^{\frac{4-2p}{p}}\left\Vert \xi\right\Vert _{p}^{2p}\left\Vert \nabla\overline{U}_{\mu}(x_{k})\right\Vert ^{2}\mid x_{k}\right]\\
 & \leq d^{\frac{4-2p}{p}}E\left[\left\Vert \xi\right\Vert _{p}^{2p}\right]\left\Vert \nabla\overline{U}_{\mu}(x_{k})\right\Vert ^{2}\\
 & \leq d^{\frac{4-2p}{p}}d\left(d+1\right)\left\Vert \nabla\overline{U}_{\mu}(x_{k})\right\Vert ^{2}\\
 & \leq\left[d+1\right]^{\frac{4}{p}}\left\Vert \nabla\overline{U}_{\mu}(x_{k})\right\Vert ^{2}.
\end{align*}

\begin{align*}
E[\|B-C\|^{2}\mid x_{k}] & \leq2E[\left\Vert B\right\Vert ^{2}\mid x_{k}]+2E[\left\Vert C\right\Vert ^{2}\mid x_{k}]\\
 & =2\left[d+1\right]^{\frac{4}{p}}\left\Vert \nabla\overline{U}_{\mu}(x_{k})\right\Vert ^{2}+2\left\Vert \nabla\overline{U}_{\mu}(x_{k})\right\Vert ^{2}\\
 & \leq2\left[d+2\right]^{\frac{4}{p}}\left\Vert \nabla\overline{U}_{\mu}(x_{k})\right\Vert ^{2}
\end{align*}

Thus the variance is bounded by 
\begin{align*}
E\left[\left\Vert \zeta_{k}-E[\zeta_{k}\mid x_{k}]\right\Vert ^{2}\right] & \leq\left(E\left[\left\Vert A\right\Vert ^{2}\right]{}^{\frac{1}{2}}+E\left[\left\Vert B-C\right\Vert {}^{2}\right]^{\frac{1}{2}}\right)^{2}\\
 & \leq\left(\frac{1}{2}(M+\lambda)\mu\left(d+3\right){}^{3\slash p}+\sqrt{2}\left(d+2\right)^{\frac{2}{p}}\left\Vert \nabla\overline{U}_{\mu}(x_{k})\right\Vert \right)^{2}.
\end{align*}
 \end{proof} 
\subsection{Proof of Lemma 3}
\begin{proof} This proof adapts the technique of the proof of \citep{dalalyan2019bounding}
Proposition 1. Without loss of generality we may assume that ${\displaystyle \int_{\mathbb{R}^{p}}\exp(-\overline{U}(x))dx=1}$.
We first give upper and lower bounds to the normalizing constant of
$\overline{\pi}_{\mu}$, that is 
\[
c_{\mu}\stackrel{_{\triangle}}{=}\int_{\mathbb{R}^{d}}\overline{\pi}(x)e^{-\left|\overline{U}_{\mu}(\mathrm{x})-\overline{U}(\mathrm{x})\right|}dx.
\]
The constant $c_{\mu}$ is an expectation with respect to the density
$\overline{\pi}_{\mu}$, it can be trivially upper bounded by 1. From
Lemma \ref{2.1}, $\left|\overline{U}_{\mu}(\mathrm{x})-\overline{U}(\mathrm{x})\right|=\left|E_{\xi}\left(U_{\mu}(\mathrm{x})-U(\mathrm{x})+\frac{\lambda}{2}\left\Vert x+\mu\xi\right\Vert ^{2}-\frac{\lambda}{2}\left\Vert x\right\Vert ^{2}\right)\right|=\left|E_{\xi}\left(U_{\mu}(\mathrm{x})-U(\mathrm{x})\right)\right|+\frac{\lambda}{2}\mu^{2}E\left\Vert \xi\right\Vert ^{2}\leq a=\frac{L\mu^{1+\alpha}d^{\frac{1+\alpha}{p}}}{1+\alpha}+\frac{1}{2}\lambda\mu^{2}\left(d+1\right)^{\frac{2}{p}}$,
where we have used $E\left\Vert \xi\right\Vert ^{2}\leq\left(d+1\right)^{\frac{2}{p}}$ from Lemma
4.
So, we have $e^{-\left|\overline{U}_{\mu}(\mathrm{x})-\overline{U}(\mathrm{x})\right|}\geq e^{-a}.$
This fact yields $e^{-a}\leq c_{\mu}\leq1.$ Now we control the distance
between densities $\overline{\pi}$ and $\overline{\pi}_{\mu}$ at
any fixed $x\in\mathbb{R}^{d}$: 
\begin{align*}
\left|\overline{\pi}(x)-\overline{\pi}_{\mu}(x)\right| & =\overline{\pi}(x)\left|1-\frac{e^{-\left|\overline{U}_{\mu}(\mathrm{x})-\overline{U}(\mathrm{x})\right|}}{c_{\mu}}\right|\\
~ & \leq\overline{\pi}(x)\left\{ \left(1-e^{-\left|\overline{U}_{\mu}(\mathrm{x})-\overline{U}(\mathrm{x})\right|}\right)+e^{-\left|\overline{U}_{\mu}(\mathrm{x})-\overline{U}(\mathrm{x})\right|}\left(\frac{1}{c_{\mu}}-1\right)\right\} \\
 & \leq\overline{\pi}(x)\left(\left|\overline{U}_{\mu}(\mathrm{x})-\overline{U}(\mathrm{x})\right|+e^{a}-1\right).
\end{align*}
The second inequality is trivial while the last inequality follows
from $1-e^{-x}\leq x$ for any $x\geq0$. To bound $W_{2}$, we use
an inequality from \citep{villani2008optimal}(Theorem 6.15, page
115): 
\[
W_{2}^{2}(\mu,\ \nu)\leq2\int_{\mathbb{R}^{d}}\Vert x\Vert_{2}^{2}|\mu(x)-\nu(x)|dx.
\]
Combining this with the bound on $\left|\overline{\pi}(x)-\overline{\pi}_{\mu}(x)\right|$
shown above, we have 
\[
W_{2}^{2}(\overline{\pi},\ \overline{\pi}_{\mu})\leq2\int_{\mathbb{R}^{d}}\Vert x\Vert_{2}^{2}\overline{\pi}(x)\left(a+e^{a}-1\right)dx.
\]
By \citep{durmus2019high}, the following bound on the second moment,
centered on the mode holds 
\[
\int_{\mathbb{R}^{d}}\Vert x-x^{*}\Vert_{2}^{2}\overline{\pi}(x)dx\leq\frac{d}{\lambda}.
\]
In addition, by Young inequality 
\[
\int_{\mathbb{R}^{d}}\Vert x\Vert_{2}^{2}\overline{\pi}(x)dx\leq2\int_{\mathbb{R}^{d}}\Vert x-x^{*}\Vert_{2}^{2}\overline{\pi}(x)dx+2\int_{\mathbb{R}^{d}}\Vert x^{*}\Vert_{2}^{2}\overline{\pi}(x)dx
\]
\[
\leq\frac{2d}{\lambda}+2\Vert x^{*}\Vert_{2}^{2}.
\]
Therefore 
\[
W_{2}^{2}(\overline{\pi},\ \overline{\pi}_{\mu})\leq\frac{4(d+\lambda\Vert x^{*}\Vert_{2}^{2})}{\lambda}\left(a+e^{a}-1\right),
\]
which is the claim of the Lemma. Finally, $a<0.1$ ensures that $e^{a}-1\leq1.06a$.
Follow from previous inequality we have, 
\[
W_{2}(\overline{\pi},\ \overline{\pi}_{\mu})\leq\sqrt{\frac{8.24(d+\lambda\Vert x^{*}\Vert_{2}^{2})a}{\lambda}}.
\]
Treating $L,\,\mu,\,\Vert x^{*}\Vert_{2}^{2}$ as constants, choose
$\lambda$ small enough so that $8.24\lambda\Vert x^{*}\Vert_{2}^{2}<0.76d$,
then $W_{2}(\overline{\pi},\ \overline{\pi}_{\mu})$ is less than
$3\sqrt{\frac{da}{\lambda}}$. 
\end{proof}
\subsection{Additional lemma }
\begin{lem}\label{L4}
If $\xi\sim N_{p}\left(0,I_{d}\right)$ then $d^{\left\lfloor \frac{n}{p}\right\rfloor }\leq E(\left\Vert \xi\right\Vert _{p}^{n})\leq\left[d+\frac{n}{2}\right]^{\frac{n}{p}}$where$\left\lfloor x\right\rfloor $
denotes the largest integer less than or equal to $x.$ If $n=kp,$
then $E(\left\Vert \xi\right\Vert _{p}^{n})=d..(d+k-1)$.
\end{lem}
\begin{proof} From \citep{richter2007generalized}, we have $E(\left\Vert \xi\right\Vert _{p}^{n})=p^{\frac{n}{p}}\frac{\Gamma\left(\frac{d+n}{p}\right)}{\Gamma\left(\frac{d}{p}\right)}.$ 

Since $\Gamma$ is an inscreasing function, $p^{\frac{n}{p}}\frac{\Gamma\left(\frac{d+n}{p}\right)}{\Gamma\left(\frac{d}{p}\right)}\geq p^{\frac{n}{p}}\frac{\Gamma\left(\frac{d}{p}+\left\lfloor \frac{n}{p}\right\rfloor \right)}{\Gamma\left(\frac{d}{p}\right)}=p^{\frac{n}{p}}\frac{d}{p}\ldots\left(\frac{d}{p}+k-1\right)\geq d^{\left\lfloor \frac{n}{p}\right\rfloor }.$ 

If $n=kp$ for $k\in N$ then $E(\left\Vert \xi\right\Vert _{p}^{n})=p^{\frac{n}{p}}\frac{d}{p}\ldots\left(\frac{d}{p}+k-1\right).$
If $n\neq kp$, let $\left\lfloor \frac{n}{p}\right\rfloor =k$. Since
$\Gamma$ is log-convex, by Jensen's inequality for any $p\geq1$,
we acquire 
\begin{align*}
 & \left(1-\frac{n}{p\left\lfloor \frac{n}{p}\right\rfloor +p}\right)\log\Gamma\left(\frac{d}{p}\right)+\frac{n}{p\left\lfloor \frac{n}{p}\right\rfloor +p}\log\Gamma\left(\frac{d}{p}+\left\lfloor \frac{n}{p}\right\rfloor +1\right)\\
 & \geq\log\Gamma\left(\left(1-\frac{n}{p\left\lfloor \frac{n}{p}\right\rfloor +p}\right)\frac{d}{p}+\frac{n}{p\left\lfloor \frac{n}{p}\right\rfloor +p}\left(\frac{d}{p}+\left\lfloor \frac{n}{p}\right\rfloor +1\right)\right)\\
 & \geq\log\Gamma\left(\frac{d+n}{p}\right)>0.
\end{align*}
Raising $e$ to the power of both sides, we get 
\[
\Gamma\left(\frac{d}{p}\right)^{\left(1-\frac{n}{p\left\lfloor \frac{n}{p}\right\rfloor +p}\right)}\Gamma\left(\frac{d}{p}+\left\lfloor \frac{n}{p}\right\rfloor +1\right)^{\frac{n}{p\left\lfloor \frac{n}{p}\right\rfloor +p}}\geq\Gamma\left(\frac{d+n}{p}\right),
\]
which implies that 
\[
\begin{array}{cc}
\left[\frac{\Gamma\left(\frac{d}{p}+\left\lfloor \frac{n}{p}\right\rfloor +1\right)}{\Gamma\left(\frac{d}{p}\right)}\right]^{\frac{n}{p\left\lfloor \frac{n}{p}\right\rfloor +p}} & \geq\frac{\Gamma\left(\frac{d+n}{p}\right)}{\Gamma\left(\frac{d}{p}\right)}\\
\left[\frac{d}{p}\ldots\left(\frac{d}{p}+\left\lfloor \frac{n}{p}\right\rfloor \right)\right]^{\frac{n}{p\left\lfloor \frac{n}{p}\right\rfloor +p}} & \geq\frac{\Gamma\left(\frac{d+n}{p}\right)}{\Gamma\left(\frac{d}{p}\right)}.
\end{array}
\]
Combining with $E(\left\Vert \xi\right\Vert _{p}^{n})=p^{\frac{n}{p}}\frac{\Gamma\left(\frac{d+n}{p}\right)}{\Gamma\left(\frac{d}{p}\right)}$
gives the conclusion.
\end{proof} 




\end{document}